\documentclass[11pt]{article}
\usepackage{times}
\usepackage{latexsym}
\usepackage{epsfig,enumerate,amsmath,amsfonts,amssymb,setspace}
\usepackage{indentfirst}
\usepackage{wrapfig,lipsum}
\usepackage{epsfig}
\usepackage{graphicx}
\usepackage{amsthm}
\usepackage{algorithm}
\usepackage{algorithmic}
\usepackage[T1]{fontenc}
\usepackage{authblk}
\usepackage[algo2e,ruled,vlined]{algorithm2e}
\newtheorem{definition}{Definition}[section]
\newtheorem{theorem}{Theorem}[section]
\newtheorem{lemma}{Lemma}[section]

\newtheorem{obs}{Observation}[section]

\title{Some Algorithmic Results on Restrained Domination in Graphs}
\author[1]{Arti Pandey\thanks{artipandey2305@gmail.com}}
\author[2]{B. S. Panda\thanks{bspanda@maths.iitd.ac.in}}
\affil[1]{Department of Computer Science and Engineering, IIIT Guwahati
\newline Ambari, G. N. B. Road, Guwahati 781001, INDIA}

\affil[2]{Department of Mathematics, Indian Institute of Technology Delhi
\newline Hauz Khas, New Delhi 110016, INDIA}

\begin{document}
\maketitle
\begin{abstract}
A set $D\subseteq V$ of a graph $G=(V,E)$ is called a restrained
dominating set of $G$ if every vertex not in $D$ is adjacent to a
vertex in $D$ and to a vertex in $V \setminus D$. The
\textsc{Minimum Restrained Domination} problem is to find a
restrained dominating set of  minimum cardinality. Given a graph
$G$, and a positive integer $k$, the \textsc{Restrained Domination
Decision} problem  is to decide whether $G$ has a restrained
dominating set of cardinality a most $k$. The \textsc{Restrained
Domination Decision} problem  is known to be NP-complete for chordal
graphs. In this paper, we strengthen this NP-completeness result by
showing that the \textsc{Restrained Domination Decision} problem
remains NP-complete for doubly chordal graphs, a subclass of chordal
graphs. We also propose a polynomial time algorithm to solve the
\textsc{Minimum Restrained Domination} problem in block graphs, a
subclass of doubly chordal graphs. The \textsc{Restrained Domination
Decision} problem is also known to be NP-complete for split graphs.
We propose a polynomial time algorithm to compute a minimum
restrained dominating set of threshold graphs, a subclass of split
graphs. In addition, we also propose polynomial time algorithms to
solve the \textsc{Minimum Restrained Domination} problem in cographs
and chain graphs. Finally, we give a new improved upper bound on
the restrained domination number, cardinality of a minimum
restrained dominating set in terms of number of vertices and minimum
degree of graph. We also give a randomized algorithm to find a
restrained dominating set whose cardinality satisfy our upper bound
with a positive probability.
\end{abstract}
\vspace*{.2cm}
\section{Introduction}
For a graph $G=(V,E)$, the sets $N_{G}(v)=\{u\in V(G)
\mid uv\in E\}$ and $N_{G}[v]=N_{G}(v)\cup \{v\}$ denote the
\emph{open neighborhood} and \emph{closed neighborhood} of a vertex
v, respectively. A vertex $v$ of a graph $G$ is said to \emph{dominate} a vertex $w$ if $w \in N_{G}[v]$. A set $D\subseteq V$ is
a \emph{dominating  set} of $G$ if every vertex of $G$ is dominated by at least one vertex in $D$.
 The \textsc{Minimum Domination} problem is to find a dominating set of minimum cardinality.
Given a graph $G$, and a positive integer $k$, the \textsc{Domination Decision} problem is to decide whether $G$
has a dominating set of cardinality at most $k$.
 The \emph{domination number} of a graph $G$,
denoted by $\gamma(G),$ is the cardinality of a minimum dominating
set of $G$. The concept of domination and its variations are widely
studied as can be seen in \cite{haynes2,haynes1}.

A dominating set $D$ is called a \emph{restrained dominating set} if every vertex not in $D$ is adjacent
to some other vertex in $V \setminus D$. The \emph{restrained domination number} of a graph $G$, denoted by
$\gamma_{r}(G),$ is the cardinality of a minimum restrained dominating set of $G$. The concept of restrained domination
was introduced by Telle and Proskurowski \cite{telle} in $1997$, albeit indirectly, as a vertex partitioning problem.
The restrained domination has been widely studied, see \cite{chen,domke1,domke2,hatt1,hatt2,hatt3,hatt4,hatt5,henning,thesis,zelinka}. An application of the concept of restrained domination is that of prisoners and guards. Each vertex not in the restrained dominating set corresponds to a position of a prisoner, and every vertex in the restrained dominating set corresponds to a position of a guard. Note that position of each prisoner is observed by a guard (to effect security) while position of each prisoner is also seen by at least one other prisoner (to protect the rights of prisoners). To minimize the cost, we want to place as few guards as possible. The restrained domination problem and its decision version are as follows:\\

\noindent\underline{\textsc{Minimum Restrained Domination (MRD)} problem}
\begin{description}
  \item[Instance:] A graph $G=(V,E)$.
  \item[Solution:] A restrained dominating set $D_r$ of $G$.
  \item[Measure:] Cardinality of $D_r$.
\end{description}

\noindent\underline{\textsc{Restrained Domination Decision (RDD)} problem}
\begin{description}
  \item[Instance:] A graph $G=(V,E)$ and a positive integer $k\leq |V|$.
  \item[Question:] Does there exist a restrained dominating set $D_r$ of $G$ such that $|D_r|\leq k$?
\end{description}

In the algorithmic graph theory, we are mainly interested in the borderline between polynomial time and NP-completeness for a given graph problem. One hierarchy of graph classes is: trees $\subset$ block graphs $\subset$ doubly chordal graphs $\subset$ chordal graphs. In this hierarchy polynomial-time algorithm for restrained domination problem is known only for trees, while  it is known to be NP-complete for chordal graphs. Here we emphasize on the gap of complexity between block graphs and doubly chordal graphs. We prove that the \textsc{Restrained Domination Decision} problem is NP-complete for doubly chordal graphs and present a dynamic programming based polynomial time algorithm to compute the cardinality of a minimum restrained dominating set for block graphs. We also study the \textsc{Minimum Restrained Domination} problem on threshold graphs, cographs, and chain graphs.

 It is also interesting to see whether there exists graph classes where domination and restrained domination problems differ in complexity. The \textsc{Minimum Domination} problem is polynomial time solvable for doubly chordal graphs~\cite{andreas}, but here we prove that the \textsc{Restrained Domination Decision} problem is NP-complete for this graph class. On the other hand, we propose a graph class, where the \textsc{Minimum Restrained Domination} problem is easily solvable, but the \textsc{Domination Decision} problem is NP-complete. Next, we give a new upper bound on the restrained domination number using probabilistic approach. We also give a randomized algorithm to find a restrained dominating set of a graph whose expected cardinality satisfy the new upper bound.

 The paper is organized as follows. In Section~\ref{sec:2}, some
pertinent definitions and some preliminary results are discussed. In
Section~\ref{sec:3}, we have shown that the \textsc{Restrained Domination Decision} problem is NP-complete for doubly chordal graphs. In
Section~\ref{sec:4}, we have shown the graph classes where the
\textsc{Minimum Domination} problem and the \textsc{Minimum Restrained Domination} problem
differ in complexity. In Section~\ref{sec:5}, we proposed a dynamic programming based algorithm
 to find a minimum restrained dominating set of block graphs. 
  In Section~\ref{sec:7},
 we studied the \textsc{Minimum Restrained Domination} problem in threshold graphs. In Section~\ref{sec:8},
 we studied the \textsc{Minimum Restrained Domination} problem in cographs. In Section~\ref{sec:9}, we studied
 the \textsc{Minimum Restrained Domination} problem in chain graphs. In Section~\ref{sec:10}, we studied a new upper
  bound on the restrained domination number of a graph, and we also proposed a randomized algorithm to find a restrained
  dominating set, whose expected cardinality satisfy the new upper bound.
   In Section~\ref{sec:11}, we conclude the paper.

\section{Preliminaries}
\label{sec:2} For a graph $G = (V,E)$, the \emph{degree} of a vertex $v$ is $|N_G(v)|$
and is denoted by $d_G(v)$. If $d_G(v)=1$, then $v$ is called a
\emph{pendant vertex}. For a set $S \subseteq V$ of the graph $G=(V,E)$, the subgraph of
$G$ induced by $S$ is defined as $G[S]=(S,E_{S})$, where $E_{S}=\{xy
\in E | x, y \in S\}$. If $G[C]$,
$C\subseteq V$, is a complete subgraph of $G$, then $C$ is called a
\emph{clique} of $G$. A graph $G = (V,E)$ is said to be \emph{bipartite} if $V(G)$
can be partitioned into two disjoint sets $X$  and $Y$ such that
every edge of $G$ joins a vertex in $X$ to a vertex  in $Y$, and such a
partition $(X,Y)$ of $V$ is called a \emph{bipartition}. A bipartite
graph with bipartition $(X,Y)$ of $V$ is denoted by $G=(X,Y,E)$. A graph $G$ is
said to be a \emph{chordal graph} if every cycle in $G$ of length at
least four has a chord, i.e., an edge joining two non-consecutive
vertices of the cycle. A chordal graph $G=(V,E)$ is a \emph{split
graph} if $V$ can be partitioned into two sets $I$ and $C$ such that
$C$ is a clique and $I$ is an independent set. A
vertex $v\in V(G)$ is a \emph{simplicial vertex} of $G$ if $N_G[v]$
is a clique of $G$. An ordering $\alpha=(v_1,v_2,...,v_n)$ is a {\it
perfect elimination ordering} (PEO) of $G$ if $v_i$ is a simplicial
vertex of $G_i=G[\{v_i,v_{i+1},...,v_n\}]$ for all $i$, $1\leq i\leq
n$. We have the following characterization for chordal graphs.
\begin{theorem}[\cite{gross}]
A graph $G$ has a PEO if and only if $G$ is chordal.
\end{theorem}

A vertex $u\in N_{G}[v]$ is a \emph{maximum neighbor} of $v$ in $G$ if
$N_{G}[w]\subseteq N_{G}[u]$ for all $w \in N_{G}[v]$. A vertex $v$ in $G$ is
called \emph{doubly simplicial} if it is a simplicial vertex and it has a
maximum neighbor in $G$. An ordering
$\sigma=(v_{1},v_{2},....,v_{n})$ of $V$ is a \emph{doubly perfect
elimination ordering (DPEO)} if $v_{i}$ is a doubly simplicial vertex
in the induced subgraph $G[\{v_{i},v_{i+1},\ldots,v_{n}\}]$ for each
$i$, $1\leq i \leq n$. A graph is \emph{doubly chordal} if it admits a
doubly perfect elimination ordering (DPEO) \cite{dual}.

%

In this paper, we only consider simple connected graphs with
at least two vertices unless otherwise mentioned specifically.

We have the following straightforward observation for any restrained dominating set of a graph.
\begin{obs}
Let $G$ be a graph and $D$ be any restrained dominating set of $G$. If $P$ denotes the set of all pendant vertices in $G$, then $P\subseteq D$.
\end{obs}
\section{Restrained domination in doubly chordal graphs}
\label{sec:3} To show that the \textsc{Restrained Domination Decision} problem  is NP-complete, we need to use a well known NP-complete problem, called Exact Cover by 3-Sets (X3C) \cite{np}, which is defined as follows:

\noindent\textbf{Exact Cover By 3-Sets (X3C)} \\
\textbf{INSTANCE:} A finite set $X$ with $|X|=3q$ and a collection $\mathcal{C}$ of 3-element subsets of $X$.\\
\textbf{QUESTION:} Does $\mathcal{C}$ contain an exact cover for
$X$, that is, a subcollection $\mathcal{C'} \subseteq \mathcal{C}$
such that every element in $X$ occurs in exactly one member of
$\mathcal{C'} ?$

\begin{theorem}
The RDD problem is NP-Complete for doubly chordal graphs.
\end{theorem}
\begin{proof} Clearly, the  RDD problem is in NP.
 To show that it is NP-complete, we establish a polynomial time reduction from Exact
 Cover by 3-Sets (X3C). Let $X=\{x_{1},x_{2},....,x_{3q}\}$
 and $\mathcal{C}=\{C_{1},C_{2},....,C_{m}\}$ be an arbitrary instance of $X3C$.

 We construct the graph $G=(V,E)$ and a positive integer $k$, as in instance of the RDD problem in the following way:

 $V=\{x_{1},x_{2},\ldots,x_{3q}\} \cup \{c_{1},c_{2},\ldots,c_{m}\} \cup
\{w_{1},w_{2},\ldots,w_{q}\}\cup \{z_{1},z_{2},\ldots,z_{q}\} \cup \{r\}$,

 $E=\{x_{i}c_{j}| x_{i} \in
C_{j}, 1 \leq i \leq 3q, 1\leq j \leq m\} \cup \{c_{i}c_{j}| 1 \leq
i < j \leq m\} \cup \{rx_{i}| 1 \leq i \leq 3q\} \cup \{rc_{i}| 1
\leq i \leq m\}\cup \{rw_{i}| 1 \leq i \leq q\} \cup \{w_{i}z_{i}|
1 \leq i \leq q\}$, and $k = 2q$.

 The graph $G$ is a doubly chordal graph as  $(x_{1},...,x_{3q},c_{1},...,c_{m},z_{1},z_{2},\ldots,z_{q},w_{1},w_{2},\ldots,w_{q},r)$ is a DPEO of $G$. The construction of the graph $G=(V,E)$ associated with an instance
of $X3C$, where $X=\{x_{1},x_{2},...,x_{6}\}$ and
$\mathcal{C}=\{C_{1}=(x_{1},x_{4},x_{6}),C_{2}=(x_{1},x_{2},x_{5}),C_{3}=(x_{2},x_{3},x_{5}), C_{4}=(x_{2},x_{4},x_{6}),C_{5}=(x_{3},x_{5},x_{6})\}$
is shown in figure~\ref{fig:1}.

\begin{figure}[h!]
 \begin{center} \includegraphics[width=6cm, height=8.5cm]{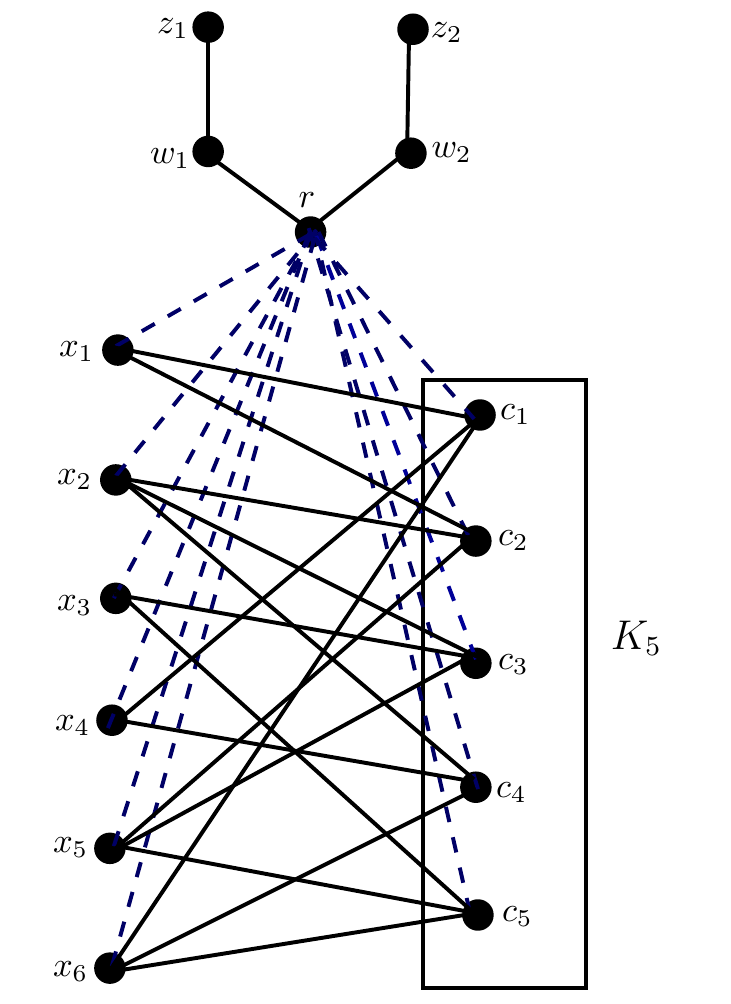}
\end{center} \caption{An Illustration to the construction of doubly chordal graph}
  \label{fig:1}
\end{figure}

Now we show that $X$ has an exact cover $\mathcal{C'}\subseteq \mathcal{C}$ if and
only if G has a restrained dominating set of cardinality at most $k$.

Suppose that $X$ has an exact cover $\mathcal{C'}$.
Then $\{c_{j}| C_{j}\in \mathcal{C'}\} \cup \{z_{1},z_{2},\ldots,z_{q}\}$ is a restrained dominating set of cardinality $2q.$

Conversely, suppose that $D$ is a restrained dominating set of $G$ of
cardinality at most $2q$. Then by Observation $2.1$, all the pendant vertices of $G$ must belong to $D$. Hence $\{z_{1},z_{2},\ldots,z_{q}\}\subseteq D$. Next, we show that $r\notin D$.
Since if $r\in D$, then the set $\{w_{1},w_{2},\ldots,w_{q}\}\subseteq D$, and the cardinality of the set $D$ must be at least $2q+1$, which is not true. Therefore $r\notin D$.

Now define $D'=D\setminus (\{z_{1},z_{2},\ldots,z_{q}\})$, and $X=\{x_{1},x_{2},\ldots,x_{3q}\}$. Then $|D'|\leq q$, and all the vertices of $X$ are dominated by $D'$. Since $N_{G}[X]=X\cup \{c_{1},c_{2},\ldots,c_{m}\}$, all the $3q$ vertices of $X$ are dominated by at most $q$ vertices of $N_{G}[X]$. If for some $i$, $1\leq i \leq n$, $x_{i}\in D$, then it dominates only a single vertex of $X$ ($x_{i}$ itself), and if for some $j$, $1\leq j\leq m$, $c_{j}\in D$, then it dominates $3$ vertices of $X$. But to dominate all the $3q$ vertices of $X$ by at most $q$ vertices of $D'$, each vertex in $D'$ must dominate at least $3$ vertices of $X$. Hence $X\cap D'=\emptyset$ and $|D'\cap \{c_{1},c_{2},\ldots,c_{m}\}|=q$. This implies that $C'=\{C_{j}\mid c_{j}\in D'\}$ is an exact cover of $\mathcal{C}$.

%

Hence, the RDD problem is NP-complete for doubly chordal graphs.
\end{proof}
\section{Complexity difference in domination and restrained domination}
\label{sec:4}
In this section, we construct a class of graphs, for which the \textsc{Minimum Restrained Domination} problem is easily solvable, but the decision version of the domination problem is NP-complete.
\begin{definition}[GP graph]
 A graph $G=(V_{G},E_{G})$ is said to be GP graph if it can be constructed from a general graph $H=(V_{H},E_{H})$, where $V_{H}=\{v_{1},v_{2},\ldots,v_{n}\}$ in the following way: for each vertex $v_{i}$ of $H$, add a path $v_{i},x_{i},y_{i},z_{i}$ of length $3$.

 Formally, $V_{G}=V_{H}\cup \{x_{i},y_{i},z_{i}\mid 1\leq i \leq n\}$ and $E_{G}=E_{H}\cup \{v_{i}x_{i},x_{i}y_{i},y_{i}z_{i}\mid 1\leq i \leq n\}$.
\end{definition}
\begin{theorem}\label{th:gp}
Let G be a GP graph constructed from a general graph $H=(V_{H},E_{H})$, where $V_{H}=\{v_{1},v_{2},\ldots,v_{n}\}$,
 by taking a path $v_{i},x_{i},y_{i},z_{i}$ of length $3$, corresponding to each vertex $v_{i}\in V_{H}$. Then
$\gamma_{r}(G)=2n$ and $V_{H}\cup \{z_{i}\mid 1\leq i \leq n\}$ is a restrained dominating set of G.
\end{theorem}
\begin{proof}
It is easy to observe that $V_{H}\cup\{z_{i}\mid 1\leq i \leq n\}$ is a restrained dominating set of $G$. Hence $\gamma_{r}(G)\leq 2n$.

Now consider a restrained dominating set, say $D_{r}$ of $G$. Then $D_{r}$ must contain all the pendant vertices of $G$. Hence $\{z_{i}\mid 1\leq i \leq n\}\subseteq D_{r}$. Now, to dominate $x_{i}$, at least one vertex from the set $\{v_{i},x_{i},y_{i}\}$ must belong to $D_{r}$, for each $i$, $1\leq i \leq n$.  This implies that $|D_{r}|\geq 2n$ and this completes the proof of the theorem.
\end{proof}

The following theorem directly follows from the Theorem~\ref{th:gp}.
\begin{theorem}
A minimum dominating set of a GP graph can be computed in linear time.
\end{theorem}

\begin{lemma}\label{prop}
Let G be a GP graph constructed from a general graph $H=(V_{H},E_H)$, where $V_{H}=\{v_{1},v_{2},\ldots,v_{n}\}$,
 by taking a path $v_{i},x_{i},y_{i},z_{i}$ of length $3$, corresponding to each vertex $v_{i}\in V_{H}$. Then
$H$ has a dominating set of cardinality at most $k$ if and only if $G$ has a dominating set of cardinality at most $n+k$.
\end{lemma}
\begin{proof}
Let $D'$ be a dominating set of $H$ of cardinality $k$. Then $D'\cup\{y_{i}\mid 1\leq i \leq n\}$ is a dominating set of $G$ of cardinality $n+k$.

Conversely, suppose that $D$ is a dominating set of $G$ of cardinality $n+k$. Then, either the pendant vertex $z_{i}$ or the vertex adjacent to pendant vertex, that is, $y_{i}$ must belong to $D$. Define $D'=D\setminus(\{y_{1},y_{2},\ldots,y_{n}\}\cup \{z_{1},z_{2},\ldots,z_{n}\})$. Then $|D'|\leq |D|-n$. Now, for each $i$, $1\leq i \leq n$, if $x_{i}\in D'$, we update $D'$ as $D'=(D'\setminus\{x_{i}\})\cup \{v_{i}\}$. Clearly $D'$ is a dominating set of $H$ and $|D'|\leq |D|-n$. Hence $D'$ is a dominating set of $H$ of cardinality at most $k$.
\end{proof}

We already have the following result for the decision version of the domination problem.
\begin{theorem}\label{domtheorem}
\cite{np} The \textsc{Domination Decision} problem is NP-complete for general graphs.
\end{theorem}
The following theorem directly follows from the Lemma~\ref{prop} and Theorem~\ref{domtheorem}.
\begin{theorem}
The \textsc{Domination Decision} problem is NP-complete for GP graphs.
\end{theorem}

\section{Restrained domination in block graphs}
\label{sec:5} \label{sec:5} A vertex $v \in V$ of a graph $G=(V,E)$
is called a \emph{cut vertex} of $G$ if $G\setminus\{v\}$, the
subgraph of $G$ obtained after removing the vertex $v$ and all the
edges incident on $v$, becomes disconnected. A maximal connected
induced subgraph with no cut vertex is called a \emph{block} of $G$.
The intersection of two blocks contains at most one vertex. A vertex
belongs to the intersection of two or more blocks if and only if it
is a cut-vertex of the graph.  A graph $G$ is called a {\em block graph}
if all the blocks of $G$ are complete graphs. A block graph can be
represented by a tree like decomposition structure, called
\emph{cut-tree}. The cut-tree, denoted by $T_{CG}(V^{CG},E^{CG})$, of a
block graph $G(V,E)$ with $k$ blocks $BC_{1}, BC_{2},\ldots, BC_{k}$
and $l$ cut vertices $v_{1},v_{2},\ldots,v_{l}$ is defined in the
following way:

$V^{CG}=\{BC_{1}, BC_{2},\ldots, BC_{k},v_{1},v_{2},\ldots,v_{l}\}$, \\
 and $E^{CG}=\{(BC_{i},v_{j})|v_{j}\in V(BC_{i}), 1\leq i \leq k, 1\leq j \leq l\}$.

The cut-tree of a block graph $G$ can be constructed in linear time
using depth-first search method. For any block $BC_{i}$,
define $B_{i}=\{v\in V(BC_{i})| v$ is not a cut vertex$\}.$ Now we can
refine the cut-tree $T_{CG}=(V^{CG},E^{CG})$ as
$T_{G}=(V^{G},E^{G})$, where $V^{G}=\{(1,B_{1}), (2,B_{2}),\ldots,
(k,B_{k}),v_{1},v_{2},\ldots,v_{l}\}$ and
$E^{G}=\{((i,B_{i}),v_{j})|v_{j}\in V(BC_{i}), 1\leq i \leq k, 1\leq j
\leq l\}$. Each $(i,B_{i})$ is called a block-vertex. Note that one
or more $B_{i}$ may be empty. So we have used $(i,B_i)$ instead of
$B_i$. However, in the rest of the paper we will use $B_i$ for
$(i,B_i)$ unless otherwise mentioned explicitly. Note that every vertex in the refined cut-tree is
 either a cut vertex or block vertex. A block graph $G$
and the corresponding refined cut-tree are shown in
Fig.~\ref{fig:2}. We consider the refined cut-tree
$T_{G}=(V^{G},E^{G})$ of the block graph $G$, as input of our
problem. Now, first we prove the following lemma.

\begin{figure}[h!]
 \begin{center} \includegraphics[width=12cm, height=5cm]{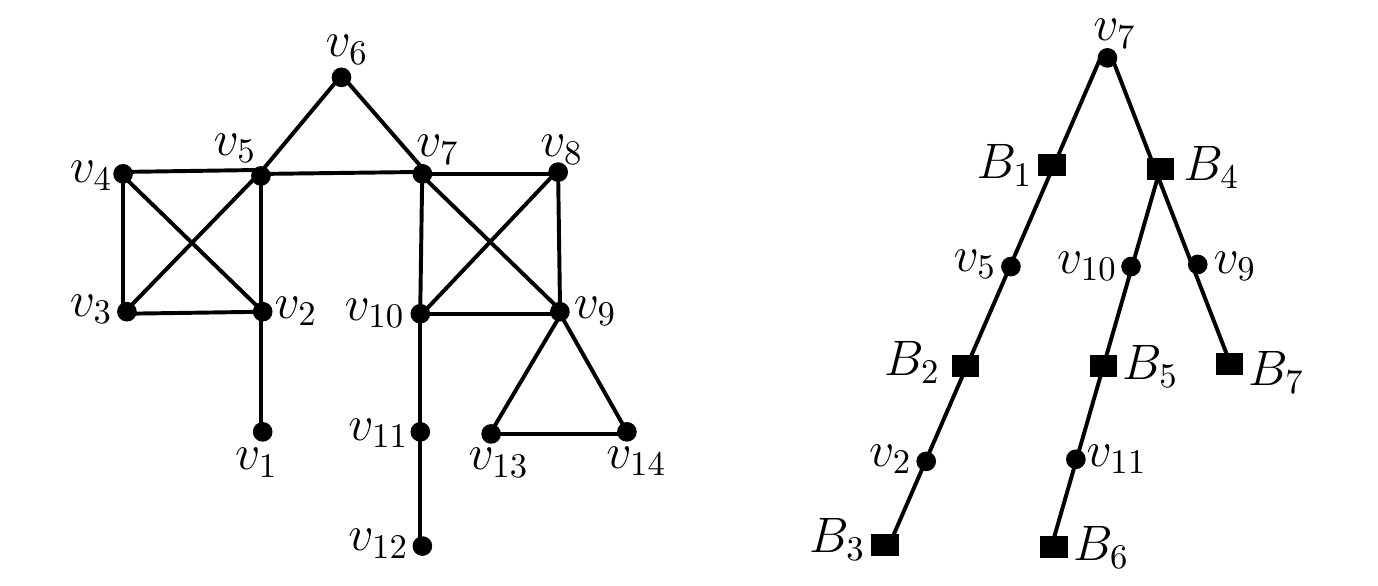}
\end{center} \caption{A block graph $G$ and the corresponding refined cut-tree}
  \label{fig:2}
\end{figure}

\begin{lemma}
Let $G=(V,E)$ be a block graph with at least three vertices. Then,
every minimum restrained dominating set  of $G$ contains at most one
vertex from each $B_{i}$.
\end{lemma}

\begin{proof}
On the contrary,  let $D$ be a minimum cardinality restrained
dominating set containing two vertices, say  $u$ and $v$,  from some
$B_{j}$ (Note that $D$  will not contain more than two vertices from
$B_{j}$; otherwise $(D \setminus B_j)\cup \{u\}$ would have been a
restrained dominating set of $G$ with smaller cardinality than that
of $D$ contradicting the minimality of $D$).

Note that $N_G[u] = N_G[v]$. If $|B_{j}| \geq 3$, then let $x \in B_j
\setminus\{u,v\}$. Now $D \setminus \{u\}$ is a restrained
dominating set of $G$ contradicting the minimality of $D$. Hence
$|B_{j}|=2$. Since $G$ has at least three vertices, block $BC_{j}$
contains one or more cut vertices.  If $D$ contains a cut vertex of
$BC_j$, then define $D'=D\setminus \{u,v\}$ else define $D'=D\setminus \{u\}$. Then
$D'$ is also a restrained dominating set of $G$, and $|D'|<|D|$, which is a contradiction to the minimality of $D$. Hence $D$
contains at most one vertex of $B_{j}$.


This completes the proof of our lemma.

\end{proof}

\subsection{Dynamic programming approach}
Let $G=(V,E)$ be a block graph with at least three vertices. If
$G=(V,E)$ has exactly one block, then $G$ is a complete graph and
$\{v\}, v \in V(G)$ is a minimum cardinality restrained dominating
set of $G$. So assume that $G$ be a connected block graph with at
least two blocks, and $T^{G}=(V^{G}, E^{G})$ be the refined cut-tree
of $G$. We make $T^{G}$ a rooted tree rooted at a cut vertex $c$ of
$G$. Consider the refined rooted cut-tree $T^{G}_{c}$ rooted at $c$,
corresponding to block graph $G$. We define the following
parameters:

$A_{c}(G)=$Min $\{|D|: c\in D$ and $D$ is a restrained dominating
set of $G\}.$

$B_{c}(G)=$Min $\{|D|: c\notin D$ and $D$ is a restrained dominating
set of $G\}.$

Since a minimum restrained dominating set of $G$ either contains the
cut vertex $c$ or does not contain the cut vertex $c$. Hence we have
the following straightforward result.

\begin{obs}
$\gamma_{r}(G)=Min (A_{c}(G), B_{c}(G))$.
\end{obs}
These parameters can be computed in a bottom up approach using the
refined cut-tree rooted at the cut vertex $c$.

\subsection{Parameters to be computed at a cut vertex in the refined cut-tree}
Let $r$ be a cut vertex in the refined cut-tree $T_{c}^{G}$ rooted at cut vertex $c$ of block graph $G$. Let $T_{r}^{G}$ be the subtree of tree $T_{c}^{G}$ rooted at cut vertex $r$. If $R$ denotes the set containing the vertex $r$ and all its descendants, then $T_{r}^{G}=T_{c}^{G}[R]$. Let $G_{r}$ denote the subgraph of $G$ reconstructed from the subtree $T_{r}^{G}$ using following construction:\\

\noindent \textbf{Construction 1:}
Let $\mathcal{C}_{i}$ denotes the set of cut vertices of the tree
$T_{r}^{G}$, and let $\mathcal{B}_{i}$ denote the set of block
vertices of the tree $T_{r}^{G}$. Then
$V(G_{r})=\mathcal{C}_{i}\cup \{B\mid B\in \mathcal{B}_{i}\}$, and
$G_{r}=G[V(G_{r})]$.\\

We compute the following parameters at every cut vertex node $r$ of the tree $T_{c}^{G}$.

$A_{r}(G_{r})=$Min $\{|D|: r\in D$ and $D$ is a restrained dominating
set of $G_{r}\}.$

$B_{r}(G_{r})=$Min $\{|D|: r\notin D$ and $D$ is a restrained dominating
set of $G_{r}\}.$

$C_{r}(G_r)=$Min $\{|D| : r \notin D, D$ dominates $V(G_{r}) \setminus \{r\}$
and every vertex $v\notin D$ has an adjacent vertex $w\notin D\}$.

$D_{r}(G_{r})=$Min $\{|D| : r \notin D, D$ dominates $V(G_{r}) \setminus \{r\}$
and every vertex $v\notin D\cup\{r\}$ has an adjacent vertex
$w\notin D\}$.

$E_{c}(G_r)=$Min $\{|D| : r \notin D, D$ dominates $V(G_{r})$ and every vertex $v\notin D\cup\{r\}$ has an adjacent vertex $w\notin D\}$.


\subsection{Parameters to be computed at a block vertex in the refined cut-tree}

Suppose $T_{c}^{G}$ is the refined cut-tree rooted at the cut vertex
$c$ and $B_{i}$ is a block vertex of $T_{c}^{G}$. Let
$T_{B_{i}}^{G}$ denote the subtree of the refined cut-tree
$T_{c}^{G}$ rooted at the vertex $B_{i}$. Also suppose that $BG_{i}$
denotes the graph reconstructed from the tree $T_{B_{i}}^{G}$ using Construction $1$.
 Note that the tree $T_{B_{i}}^{G}$ is not
necessarily the refined cut-tree of the graph $BG_{i}$. Let $BC_{i}$
denote the block corresponding to the block vertex $B_{i}$. Now we
compute the following parameters at every block vertex $B_{i}$ of tree $T_{c}^{G}$.

$A_{B_{i}}(BG_{i})=$Min $\{|D|: D$ dominates $V(BG_{i})\setminus V(BC_{i})$ and every vertex $v\in V(BG_{i})\setminus D$
has an adjacent vertex $w\in V(BG_{i})\setminus D\}$.

$B_{B_{i}}(BG_{i})=$Min $\{|D|: D$ dominates $V(BG_{i})$ and every vertex $v\in V(BG_{i})\setminus (D\cup V(BC_{i}))$ has an
adjacent vertex $w\in V(BG_{i})\setminus D\}$.

$F_{B_{i}}(BG_{i})=$Min $\{|D|: D$ dominates $V(BG_{i})$, at least one vertex of $V(BC_{i}) \cap V(BG_{i})$
does not belong to $D$,  and every vertex $v\in V(BG_{i})\setminus (D\cup V(BC_{i}))$ has an adjacent vertex $w\in V(BG_{i})\setminus D\}$.

$H_{B_{i}}(BG_{i})=$Min $\{|D|: D$ dominates $V(BG_{i})$, at least one vertex of $V(BC_{i})\cap V(BG_{i})$ belongs to $D$,
 and every vertex $v\in V(BG_{i})\setminus (D\cup V(BC_{i}))$ has an adjacent vertex $w\in V(BG_{i})\setminus D\}$.

$I_{B_{i}}(BG_{i})=$Min $\{|D|: D$ dominates $V(BG_{i})$, at least one vertex of $V(BC_{i})\cap V(BG_{i})$ does not belong to $D$,
 at least one vertex of $V(BC_{i})\cap V(BG_{i})$ belongs to $D$  and every vertex $v\in V(BG_{i})\setminus (D\cup V(BC_{i}))$ has
  an adjacent vertex $w\in V(BG_{i})\setminus D\}$.

 \textbf{If $B_{i}$ is empty, then we also define the following three parameters: }

$C_{B_{i}}(BG_{i})=$Min $\{|D|: D$ dominates $V(BG_{i})$, at least one vertex of $V(BC_{i})\cap V(BG_{i})$ does not
belong to $D$, and every vertex $v\in V(BG_{i})\setminus (D\cup V(BC_{i}))$ has an adjacent vertex $w\in V(BG_{i})\setminus D\}$.

$D_{B_{i}}(BG_{i})=$Min $\{|D|: D$ dominates $V(BG_{i})$, at least one vertex of $V(BC_{i})\cap V(BG_{i})$ belongs to $D$,
and every vertex $v\in V(BG_{i})\setminus (D\cup V(BC_{i}))$ has an adjacent vertex $w\in V(BG_{i})\setminus D\}$.

$E_{B_{i}}(BG_{i})=$Min $\{|D|: D$ dominates $V(BG_{i})$, at least one vertex of $V(BC_{i})\cap V(BG_{i})$ does
not belong to $D$, at least one vertex of $V(BC_{i})\cap V(BG_{i})$ belongs to $D$  and every vertex
$v\in V(BG_{i})\setminus (D\cup V(BC_{i}))$ has an
adjacent vertex $w\in V(BG_{i})\setminus D\}$.

\subsection{Computation of the parameters at a leaf vertex in the refined cut-tree}

\noindent {\bf Rule $1$:}\\

Let $B_{i}$ be a leaf vertex of tree $T^{G}_{c}$, and $BC_{i}$ be the corresponding block of $G$. Note that every
leaf vertex of tree is a block vertex.  In this case,
$T^{G}_{B_{i}}$ consists of only one block-vertex $B_{i}$, and $BG_{i}=G[B_{i}]$. Note that $B_{i}\neq \emptyset$ and
$V(BG_{i})\setminus V(BC_{i})=\emptyset$.
 Now, the parameters can be computed in the following way:

\[
    A_{B_{i}}(BG_{i})=
\begin{cases}
    1 ,&  if \hspace*{.2cm}|B_{i}|=1 \\
    0,           &   otherwise
\end{cases}
\]


$$B_{B_{i}}(BG_{i})=1$$.

\[
    F_{B_{i}}(BG_{i})=
\begin{cases}
    \infty ,&  if \hspace*{.2cm}|B_{i}|=1 \\
    1,           &   otherwise
\end{cases}
\]

$$H_{B_{i}}(BG_{i})=1$$.

\[
    I_{B_{i}}(BG_{i})=
\begin{cases}
    \infty ,&  if \hspace*{.2cm}|B_{i}|=1 \\
    1,           &   otherwise
\end{cases}
\]

\subsection{Computation of the parameters at a cut vertex in the refined cut-tree}
\noindent {\bf Rule $2$:}\\

Let $r$ be a cut vertex of the refined cut-tree $T^{G}_{c}$. Let $T^{G}_{r}$ be the subtree of $T^{G}_{c}$ rooted at $r$ and $G_{r}$ be the graph
reconstructed from the tree $T^{G}_{r}$, as discussed earlier.
Also suppose that the  vertex $r$ is having $k$
children, say $B_{1},B_{2},\ldots,B_{k}$, in the tree $T^{G}_{c}$. Let
$BC_{1},BC_{2},\ldots,BC_{k}$ be the blocks of $G$ corresponding to
the block vertices $B_{1},B_{2},\ldots,B_{k}$, respectively. Let $T_{B_{1}}^{G},T_{B_{2}}^{G},\ldots, T_{B_{k}}^{G}$ be the
subtrees of the tree $T_{c}^{G}$, and $BG_{1},BG_{2},\ldots,BG_{k}$ be the subgraphs of $G$ reconstructed from
the trees $T_{B_{1}}^{G},T_{B_{2}}^{G},\ldots, T_{B_{k}}^{G},$ respectively. Then, the parameters
corresponding to cut vertex $r$
can be computed in the following way:

\subsection*{Parameter $A_{r}(G_{r})$}
Suppose $D$ is a minimum cardinality restrained dominating set of
$G_{r}$ containing the vertex $r$. Then $r$ dominates all the vertices
of the blocks $BC_{1},BC_{2},\ldots,BC_{k}$. Hence, the parameter
$A_{r}(G_{r})$ can be defined in the following way:

$A_{r}(G_{r})=1+\sum_{i=1}^{k}A_{B_{i}}(BG_{i})$.

\subsection*{Parameter $B_{r}(G_{r})$}
Suppose $D$ is a minimum cardinality restrained dominating set of $G_{r}$ not containing the vertex $r$.
Since $B_{1},B_{2},\ldots,B_{k}$ are children of $r$ in the tree $T_{r}^{G}$, the cut vertex $r$ belong to $k$ blocks
$BC_{1},BC_{2},\ldots,BC_{k}$, where $k\geq 1$. Now we may have two possibilities for the set $D$.

\noindent $(i)$ There exists an index $i$, $1\leq i \leq k$, such that $V(BG_{i})\setminus D \neq \emptyset$, and $V(BG_{i})\cap D \neq \emptyset$. In
this case
$|D|=\phi=Min_{i}(I_{B_{i}}(BG_{i})+\sum_{j=1, (j \neq i)}^{k}B_{B_{j}}(BG_{j}))$.

\noindent $(ii)$ There exist indices $i,j$, $1\leq i,j \leq k$, such that $V(BG_{i})\setminus D \neq \emptyset$, and $V(BG_{j})\cap D \neq \emptyset$. In
 this case
$|D|=\psi=Min_{i,j(i\neq j)}(F_{B_{i}}(BG_{i})+H_{B_{j}}(BG_{j})+\sum_{m=1, (m\neq i,j)}^{k}B_{B_{m}}(BG_{m}))$.

Hence $B_{r}(G_{r})=Min(\phi,\psi)$.

\subsection*{Parameter $C_{r}(G_{r})$}
$C_{r}(G_{r})$ can be computed as follows.

$C_{r}(G_{r})=Min_{i}(F_{B_{i}}(BG_{i})+\sum_{j=1,(j\neq
i)}^{k}B_{B_{j}}(BG_{j}))$.

\subsection*{Parameter $D_{r}(G_{r})$}
Let $D$ be a minimum cardinality set not containing $r$ such that
$D$ dominates all the vertices of $G_{r}$ except $r$, and every vertex
$v\notin D\cup \{r\}$ has an adjacent vertex $w \notin D$. Now, each
vertex in the block $V(BC_{i})$, for all $i$, $1\leq i \leq k$, is
adjacent to the vertex $r$, which is not in $D$. Hence $D_{r}(G_{r})$ is
equal to the cardinality of a minimum cardinality set $D$ not
containing $r$ such that $D$ dominates all the vertices of $G_{r}$
except $r$, and every vertex $v\notin D \cup (\cup_{i=1}^{k}BC_{i})$
has an adjacent vertex $w\notin D$. So, the parameter $D_{r}(G_{r})$ can
be defined in the following way:

$D_{r}(G_{r})=\sum_{i=1}^{k}B_{B_{i}}(BG_{i})$.

\subsection*{Parameter $E_{r}(G_{r})$}
Here we have two cases to consider.

If at least one of the $B_{i}$ is non-empty, then

$E_{r}(G_{r})=\sum_{i=1}^{k}{B_{B_i}}(BG_{i})$, \\ otherwise, if all the $B_{i}'s$ are empty, then

$E_{r}(G_{r})=Min_{i}(D_{B_{i}}(BG_{i})+\sum_{j=1, (j\neq i)}^{k}B_{B_{j}}(BG_{j}))$.

\subsection{Computation of the parameters at a non-leaf block vertex in the refined cut-tree}
\noindent {\bf Rule $3$:}\\

Let $B_{k}$ be a non-leaf block vertex of tree $T_{c}^{G}$ and $BC_{k}$ be the corresponding block of $G$.
Let $T_{B_{k}}^{G}$ denote the subtree of the tree
$T_{c}^{G}$ rooted at the vertex $B_{k}$. Let $BG_{k}$
denote the graph reconstructed from the tree $T_{B_{k}}^{G}$ as discussed earlier. Also suppose that
 $x_{1},x_{2},\ldots,x_{p}$ are children of $B_{k}$
in the tree  $T_{c}^{G}$. Let $T_{x_{1}}^{G},T_{x_{2}}^{G},\ldots,T_{x_{p}}^{G}$ be the subtrees of the tree
$T_{c}^{G}$ rooted at the vertices  $x_{1},x_{2},\ldots,x_{p}$
 respectively, and $G_{x_{1}},G_{x_{2}},\ldots,G_{x_{p}}$ be subgraphs of $G$ reconstructed from the trees
 $T_{x_{1}}^{G},T_{x_{2}}^{G},\ldots,T_{x_{p}}^{G}$, respectively. Then, the parameters corresponding to block
 vertex $B_{k}$ can be computed in the following way:

\subsubsection*{Parameter $A_{B_{k}}(BG_{k})$}
Let $D$ be a minimum cardinality subset of $V(BG_{i})$ such that
$D$ dominates $V(BG_{i})\setminus V(BC_{i})$, and every vertex
 $v\in V(BG_{i})\setminus D$ has an adjacent vertex $w\in V(BG_{i})\setminus D$.
  Here we have three cases to consider depending on the cardinality of the set $B_{k}$.

 \textbf{Case 1:} $|B_{k}|\geq 2$.

In this case, we have  $A_{B_{k}}(BG_{k})=\sum_{i=1}^{p}
Min(A_{x_{i}}(G_{x_{i}}),D_{x_{i}}(G_{x_{i}}))$.

\textbf{Case 2:} $|B_{k}|= 1$.

Again we have following subcases to consider.

Subcase $2.1$ : $|B_{k}|=1$ and the vertex of $B_{k}$ does not belong to $D$.

Here $|D|=D1=Min_{i}(D_{x_{i}}(G_{x_{i}})+\sum_{j=1, (j\neq i)}^{p}Min(A_{x_{j}}(G_{x_{j}}),D_{x_{j}}(G_{x_{j}}))$.

Subcase $2.2$ :  $|B_{k}|=1$ and the vertex of $B_{k}$ belongs to $D$.

This is possible only if all the children of $B_{k}$ belong to $D$, that is, $\{x_{1},x_{2},\ldots,x_{p}\}\subseteq D$.

Here $|D|=D2=1+\sum_{i=1}^{p}A_{x_{i}}(G_{x_{i}})$.

Hence, in this case $A_{B_{k}}(BG_{k})=Min(D1,D2)$.

\textbf{Case 3:} $|B_{k}|= 0$.

Again we have following subcases to consider.

Subcase 3.1 : $B_{k}$ has exactly one child, say, $x_{1}$.

$A_{B_{k}}(BG_{k})=Min(A_{x_{1}}(G_{x_{1}}),C_{x_{1}}(G_{x_{1}})).$

Subcase 3.2 : $B_{k}$ has two or more child. Again there are three possibilities:

(i) All the children of $B_{k}$ in the tree $T_{B_{k}}^{G}$ belong to $D$.

Here $|D|=D1=\sum_{i=1}^{p}A_{x_{i}}(G_{x_{i}})$.

(ii) Exactly one child of $B_{k}$ does not belong to $D$.

Here $|D|=D2=Min_{i}(C_{x_{i}}(G_{x_{i}})+\sum_{j=1, (j\neq i)}^{p}A_{x_{j}}(G_{x_{j}}))$.

(iii) At least two child of $B_{k}$ do not belong to $D$.

Here $|D|=D3=Min_{i,j (i\neq j)}(D_{x_{i}}(G_{x_{i}})+D_{x_{j}}(G_{x_{j}})+\sum_{k=1, (k\neq i,j)}^{p}Min(A_{x_{k}}(G_{x_{k}}),D_{x_{k}}(G_{x_{k}})))$.

Hence, in this subcase $A_{B_{k}}(BG_{k})=Min(D1,D2,D3)$.

\subsubsection*{Parameter $B_{B_{k}}(BG_{k})$}
Let $D$ be the minimum cardinality subset of the set $V(BG_{k})$ such that $D$
dominates $V(BG_{k})$, and every vertex $v\in V(BG_{k})\setminus (D\cup BC_{k})$ has an adjacent
vertex $w \in V(BG_{k})\setminus D$.
Here we have two cases to consider depending on the cardinality of the set $B_{k}$.

\textbf{Case 1:} $B_{k}$ is non-empty.

Again there are two possibilities:

(i) A vertex $x\in B_{k}$ belongs to $D$.

Here $|D|=D1=1+\sum_{i=1}^{p} Min(A_{x_{i}}(G_{x_{i}}),D_{x_{i}}(G_{x_{i}}))$.

(ii) No vertex of $B_{k}$ belongs to $D$.

Here $|D|=D2=Min_{i}(A_{x_{i}}(G_{x_{i}})+\sum_{j=1, (j\neq i)}^{p}Min(A_{x_{j}}(G_{x_{j}}),D_{x_{j}}BG_{x_{j}}))$.

Hence, in this case $B_{B_{k}}(BG_{k})=Min(D1,D2)$.

\textbf{Case 2:} $B_{k}$ is empty.

Again there are two possibilities to consider.

(i) No child of $B_{k}$ belongs to $D$ (that is, $\{x_{1},x_{2},\ldots,x_{p}\}\cap D=\emptyset$).

Here $|D|=D1=\sum_{i=1}^{p}E_{x_{i}}(G_{x_{i}})$.

(ii) At least one child of $B_{k}$ belongs to $D$ (that is, $\{x_{1},x_{2},\ldots,x_{p}\}\cap D\neq \emptyset$).

Here $|D|=D2=Min_{i}(A_{x_{i}}(G_{x_{i}})+\sum_{j=1, (j\neq i)}^{p}Min(A_{x_{j}}(G_{x_{j}}),D_{x_{j}}(G_{x_{j}}))$.

Hence, in this case $B_{B_{k}}(BG_{k})=Min(D1,D2)$.

\subsubsection*{Parameter $E_{B_{k}}(BG_{k})$}
Note that this parameter is computed only if $B_{k}$ is an empty set.

If $B_{k}$ has exactly one child in the tree $T_{B_{k}}^{G}$, that is, $p=1$, then $E_{B_{k}}(BG_{k})=\infty$.

Otherwise $E_{B_{k}}(BG_{k})=Min_{i,j}(A_{x_{i}}(G_{x_{i}})+D_{x_{j}}(G_{x_{j}})+\sum_{m=1, (m\neq i,j)}^{p}Min(A_{x_{m}}(G_{x_{m}}),D_{x_{m}}(G_{x_{m}})))$.

\subsubsection*{Parameter $D_{B_{k}}(BG_{k})$}
Note that this parameter is computed only if $B_{k}$ is an empty set.

$D_{B_{k}}(BG_{k})=Min_{i}(A_{x_{i}}(G_{x_{i}})+\sum_{j=1, (j\neq i)}^{p}Min(A_{x_{j}}(G_{x_{j}}),D_{x_{j}}(G_{x_{j}}))$.

\subsubsection*{Parameter $C_{B_{k}}(BG_{k})$}
Note that this parameter is computed only if $B_{k}$ is an empty set.
Let $D$ be a minimum cardinality subset of $V(BG_{k})$ such that $D$ dominates $V(BG_{k})$, at least one vertex of $BC_{k}\cap V(BG_{k})$
does not belong to $D$ and every vertex $v\in V(BG_{k})\setminus (D \cup BC_{k})$ has an adjacent vertex $w\in V(BG_{k})\setminus D$.
Here we have two cases to consider.

\textbf{Case 1:} No child of $B_{k}$ belong to $D$.

Here $|D|=D1=\sum_{i=1}^{p}E_{x_{i}}(G_{x_{i}})$.

\textbf{Case 2:} At least one child of $B_{k}$ belongs to $D$.

Here $|D|=D2=E_{B_{k}}(BG_{k})$.

Hence $C_{B_{k}}(BG_{k})=Min(D1,D2)$.

\subsubsection*{Parameter $F_{B_{k}}(BG_{k})$}
Since $B_{k}$ is not a leaf vertex of tree $T_{c}^{G}$, $|B_{k}|\neq |BG_{k}|$.
Now we have two cases to consider depending on the cardinality of $B_{k}$.

\textbf{Case 1:} $B_{k}$ is non-empty.

Here $F_{B_{k}}(BG_{k})=B_{B_{k}}(BG_{k})$.

\textbf{Case 2:} $B_{k}$ is empty.

Here $F_{B_{k}}(BG_{k})=C_{B_{k}}(BG_{k})$.

\subsubsection*{Parameter $H_{B_{k}}(BG_{k})$}
Since $B_{k}$ is not a leaf vertex of tree $T_{c}^{G}$, $|B_{k}|\neq |BG_{k}|$.
Now we have two cases to consider depending on the cardinality of $B_{k}$.

\textbf{Case 1:}  $B_{k}$ is non-empty.

$H_{B_{k}}(BG_{k})=B_{B_{k}}(BG_{k})$.

\textbf{Case 2:} $B_{k}$ is empty.

$I_{B_{k}}(BG_{k})=D_{B_{k}}(BG_{k})$.

\subsubsection*{Parameter $I_{B_{k}}(BG_{k})$}

We have two cases to consider depending on the cardinality of $B_{k}$.

\textbf{Case 1:}  $B_{k}$ is non-empty.

Here $I_{B_{k}}(BG_{k})=B_{B_{k}}(BG_{k})$.

\textbf{Case 2:} $B_{k}$ is empty.

Here $I_{B_{k}}(BG_{k})=E_{B_{k}}(BG_{k})$.

\subsection{Algorithm}
We are now ready to propose an algorithm to compute $\gamma_{r}(G)$
of a block graph $G=(V,E)$ having
at least two blocks.

\begin{algorithm}[H]
\caption{Algorithm-RD(G)}
 \textbf{Input:} A block graph $G=(V,E)$ with at least two blocks.\\
\textbf{Output:} $\gamma_r(G)$.\\
\Begin{
Compute the refined cut-tree $T$ of $G$ rooted a cut vertex $c$
of $G$;\\
Compute a reverse BFS ordering $\alpha=(v_1,v_2,\ldots,v_n=c)$ of
$T$;\\
\For {i=1 to n }{
\If{$v_{i}$ is a leaf vertex of $T$}{
Compute the parameters at $v_{i}$ using Rule $1$;
}
\ElseIf{$v_{i}$ is a cut vertex of $T$}
{
Compute the parameters at $v_{i}$ using Rule $2$;
}
\ElseIf{$v_{i}$ is a non-leaf block vertex of $T$}
{
Compute the parameters at $v_{i}$ using Rule $3$;
}
}
$\gamma_{r}(G)=Min (A_{c}(G), B_{c}(G))$ for the root $c$ of $T$.\\
return $\gamma_{r}(G)$;
}
\end{algorithm}
%
%
%
%
%
%
%
%
%
%
%
%
%
%
%
%
%



\subsection{Illustration of the algorithm}

We now illustrate our algorithm using an example. Fig.~\ref{fig:2}
shows an example of block graph $G$ having seven blocks,
$BC_{1}=G[\{v_{5},v_{6},v_{7}\}]$,
$BC_{2}=G[\{v_{2},v_{3},v_{4},v_{5}\}]$,
$BC_{3}=G[\{v_{1},v_{2}\}]$,
$BC_{4}=G[\{v_{7},v_{8},v_{9},v_{10}\}]$,
$BC_{5}=G[\{v_{10},v_{11}\}]$, $BC_{6}=G[\{v_{11},v_{12}\}]$,
$BC_{7}=G[\{v_{9},v_{13},v_{14}\}]$ and the corresponding refined
cut-tree rooted at vertex $v_{6}$, say $T^{G}_{v_{6}}$. The cut-tree
$T^{G}_{v_{6}}$ is having $13$ vertices, in which six are cut
vertices $v_{2},v_{5},v_{7},v_{9},v_{10},v_{11}$ and seven are block
vertices $B_{1}=\{v_{6}\}$, $B_{2}=\{v_{3},v_{4}\}$,
$B_{3}=\{v_{1}\}$, $B_{4}=\{v_{8}\}$, $B_{5}=\emptyset$,
$B_{6}=\{v_{12}\}$, $B_{7}=\{v_{13},v_{14}\}$.

For the refined cut-tree of $G$ shown in Fig~\ref{fig:2}, the reverse of BFS
ordering is:

$\alpha=(B_{6},B_{3},v_{11},v_{2},B_{7},B_{5},B_{2},v_{9},v_{10},v_{5},B_{4},B_{1},v_{7})$.

\begin{enumerate}
\item[(1)] $B_{6}$ is a block vertex, $BC_{6}$ is an end block of $G$, $|B_{6}|=1$, $|V(BC_{6})|=2$.\\
$A_{B_{6}}=1$, $B_{B_{6}}=1$, $F_{B_{6}}=\infty$, $H_{B_{6}}=1$, $I_{B_{6}}=\infty$.
\item[(2)] $B_{3}$ is a block vertex, $BC_{3}$ is an end block of $G$, $|B_{3}|=1$, $|V(BC_{3})|=2$.\\
$A_{B_{3}}=1$, $B_{B_{3}}=1$, $F_{B_{3}}=\infty$, $H_{B_{3}}=1$, $I_{B_{3}}=\infty$.
\item[(3)] $v_{11}$ is a cut vertex having only one child $B_{6}$ in the tree $T^{G}_{v_{6}}$ and $|B_{6}|=1$.\\
$A_{v_{11}}=1+A_{B_{6}}=2$, $B_{v_{11}}=I_{B_{6}}=\infty$,  $C_{v_{11}}=F_{B_{6}}=\infty$, $D_{v_{11}}=B_{B_{6}}=1$, $E_{v_{11}}=B_{B_{6}}=1$.
\item[(4)] $v_{2}$ is a cut vertex having only one child $B_{3}$ in the tree $T^{G}_{v_{6}}$ and $|B_{3}|=1$.\\
$A_{v_{2}}=1+A_{B_{3}}=2$, $B_{v_{2}}=I_{B_{3}}=\infty$, $C_{v_{2}}=F_{B_{3}}=\infty$, $D_{v_{2}}=B_{B_{3}}=1$ $E_{v_{2}}=B_{B_{3}}=1$.
\item[(5)] $B_{7}$ is a block vertex, $BC_{7}$ is an end block of $G$, $|B_{7}|=2$, $|V(BC_{7})|=3$.\\
$A_{B_{7}}=0$, $B_{B_{7}}=1$, $F_{B_{7}}=1$, $H_{B_{7}}=1$, $I_{B_{7}}=1$.
\item[(6)] $B_{5}$ is a block vertex, $BC_{5}$ is a block of $G$, $|B_{5}|=0$, $|V(BC_{5})|=2$, and $B_{5}$ is having one child $v_{11}$ in the tree $T^{G}_{v_{6}}$.\\
$A_{B_{5}}=Min(A_{v_{11}},C_{v_{11}})=2$, $B_{B_{5}}=Min(E_{v_{11}},A_{v_{11}})=1$, $E_{B_{5}}=\infty$, $D_{B_{5}}=A_{v_{11}}=2$, $C_{B_{5}}=Min(E_{v_{11}},E_{B_{5}})=1$, $F_{B_{5}}=C_{B_{5}}=1$, $H_{B_{5}}=D_{B_{5}}=2$, $I_{B_{5}}=E_{B_{5}}=\infty$.
\item[(7)] $B_{2}$ is a block vertex, $BC_{2}$ is a block of $G$, $|B_{2}|=2$, $|V(BC_{2})|=4$, and $B_{2}$ is having one child $v_{1}$ in the tree $T^{G}_{v_{6}}$.\\
$A_{B_{2}}=Min(B_{v_{2}},1+A_{v_{2}})=1$, $B_{B_{2}}=Min(1+Min(A_{v_{2}},B_{v_{2}}),A_{v_{2}})=2$, $F_{B_{2}}=2$, $H_{B_{2}}=2$, $I_{B_{2}}=2$.
\item[(8)] $v_{9}$ is a cut vertex having only one child $B_{7}$ in the tree $T^{G}_{v_{6}}$ and $|B_{7}|=2$.\\
$A_{v_{9}}=1+A_{B_{7}}=1$, $B_{v_{9}}=I_{B_{7}}=1$, $C_{v_{9}}=F_{B_{7}}=1$, $D_{v_{9}}=B_{B_{7}}=1$,  $E_{v_{9}}=B_{B_{7}}=1$.

\item[(9)] $v_{10}$ is a cut vertex having only one child $B_{5}$ in the tree $T^{G}_{v_{6}}$ and $|B_{5}|=0$.\\
$A_{v_{10}}=1+A_{B_{5}}=3$, $B_{v_{10}}=I_{B_{5}}=\infty$, $C_{v_{10}}=F_{B_{5}}=1$, $D_{v_{10}}=B_{B_{5}}=1$, $E_{v_{10}}=D_{B_{5}}=2$.

\item[(10)] $v_{5}$ is a cut vertex having only one child $B_{2}$ in the tree $T^{G}_{v_{6}}$ and $|B_{2}|=2$.\\
$A_{v_{5}}=1+A_{B_{2}}=2$,  $B_{v_{5}}=I_{B_{2}}=2$, $C_{v_{5}}=F_{B_{2}}=2$, $D_{v_{5}}=B_{B_{2}}=2$, $E_{v_{5}}=B_{B_{2}}=2$.

\item[(11)] $B_{4}$ is a block vertex, $BC_{4}$ is a block of $G$, $|B_{4}|=1$, $|V(BC_{4})|=4$, and $B_{4}$ is having two child $v_{10}$ and $v_{9}$ in the tree $T^{G}_{v_{6}}$.\\
$A_{B_{4}}=Min(D1,D2)$ where $D1=Min(B_{v_{10}}+Min(A_{v_{9}},B_{v_{9}}),B_{v_{9}}+Min(A_{v_{10}},B_{v_{10}}))=2$ and $D2=1+A_{v_{9}}+A_{v_{10}}=5$, So $A_{B_{4}}=2$. \\
$B_{B_{4}}=Min(D_{1},D_{2})$ where $D1=1+Min(A_{v_{9}},B_{v_{9}})+ Min(A_{v_{10}},B_{v_{10}})=3$ and $D2=Min(A_{v_{9}}+Min(A_{v_{10}},B_{v_{10}}),A_{v_{10}}+Min(A_{v_{9}},B_{v_{9}}))=2$, So $B_{B_{4}}=2$. \\
$F_{B_{4}}=2$, $H_{B_{4}}=2$, $I_{B_{4}}=2$.

\item[(12)] $B_{1}$ is a block vertex, $BC_{1}$ is a block of $G$, $|B_{1}|=1$, $|V(BC_{1})|=3$, and $B_{1}$ is having one child $v_{5}$ in the tree $T^{G}_{v_{6}}$.\\
$A_{B_{1}}=Min(B_{v_{5}},1+A_{v_{5}})=2$, $B_{B_{1}}=Min(1+Min(A_{v_{5}},B_{v_{5}}),A_{v_{5}})=2$, $F_{B_{1}}=2$, $H_{B_{1}}=2$, $I_{B_{1}}=2$.

\item[(13)] $v_{7}$ is a cut vertex having two child $B_{1}$ and $B_{4}$ in the tree $T^{G}_{v_{6}}$ and $|B_{1}|=|B_{4}|=1$.\\
$A_{v_{7}}=1+A_{B_{1}}+A_{B_{4}}=5$, $B_{v_{7}}=Min(I_{B_{1}}+B_{B_{4}},I_{B_{4}}+B_{B_{1}},F_{B_{1}}+H_{B_{4}},F_{B_{4}}+H_{B_{1}})=4$, $C_{v_{7}}=Min(F_{B_{1}}+B_{B_{4}},F_{B_{4}}+B_{B_{1}})=4$,
$D_{v_{7}}=B_{B_{1}}+B_{B_{4}}=4$, $E_{v_{7}}=B_{B_{1}}+B_{B_{4}}=4$.

\end{enumerate}

Thus the minimum restrained domination number $\gamma_{r}(G)=Min(A_{v_{7}},D_{v_{7}})=4$.

\subsection{Complexity details}
\begin{theorem}
The restrained domination number of a block graph $G$ can be computed in $O(n^{4})$ time, where $n$
denotes the number of vertices in $G$.
\end{theorem}
\begin{proof}
The proof of correctness of Algorithm\_RD follows from the
recurrence relations obtained in the earlier part of this section.

Next we analyze the running time of the algorithm.

 Let $G=(V,E)$ be
a block graph with $n$ vertices and $m$ edges.
 The refined cut-tree $T$ of $G$ can be computed in $O(n+m)$ time using depth-first search
 similar to the use of depth-first search in constructing a cut-tree of $G$.
 Note that the number of vertices in the tree $T$ is also $O(n)$.
 While constructing the tree $T$, we can also maintain the following information
  for each vertex $v$ of the refined cut-tree: $(i)$ $v$ is a cut vertex or block vertex, $(ii)$ if $v$ is a block vertex,
  then the number of vertices in the corresponding block, and $(iii)$ if $v$ is a block vertex,
  then number of non-cut vertices in the corresponding block. All these information can be maintained in linear time.

Next we find a reverse of BFS ordering, say $\alpha$, of the vertices of the refined
cut-tree $T$. The ordering $\alpha$ can be computed in $O(n)$ time. Now we process the vertices
of $T$ in the ordering $\alpha$, and for each vertex we compute some parameters using Rules $1$, $2$, and $3$,
depending on the case whether it is a leaf vertex (block vertex), non-leaf block vertex, or a cut vertex.
If $v$ is a block vertex and leaf of tree $T$, then we process it according to Rule $1$, and
all the parameters can be computed in $O(1)$ time. If $v$ is a non-leaf block vertex,
we process it according to Rule $3$. If $v$ is a cut vertex, we process it according to Rule $2$.

 Let $T_{1}(n)$ be the time required to compute all the parameters
 for any cut vertex given all the computed information for its children.
 Similarly, let $T_{2}(n)$ be the time required to compute all the parameters
 for any non-leaf block vertex given all the computed information for its children.
  Then the time complexity of
 the algorithm is $T(n)=O(n+m)+O(nT_{1}(n)+nT_{2}(n))$. Now we only need to find $T_{1}(n)$ and $T_{2}(n)$.

 Let $c(v)$ denote the number of children of $v$ in tree $T$, then $d_{T}(v)=c(v)+1$.

 \noindent\textbf{Estimate for $T_{1}(n):$}\\
 From the formula for $A_{v},B_{v},C_{v},D_{v},$ and $E_{v}$, it is easy to note that
  $A_{v}$ can be computed in $O(c(v))=O(d_{T}(v))$ time,
  $B_{v}$ can also be computed in $O(d_{T}(v))$ time,
  $C_{v}$ can be computed in $O(d_{T}(v)^{2})$ time, $D_{v}$ can be computed in $O(d_{T}(v)^{3})$ time,
  and  $E_{v}$ can be computed in $O(d_{T}(v)^{2})$ time. Hence $T_{1}(n)=O(d_{T}(v)^{3})=O(n^{3})$.

 \noindent\textbf{Estimate for $T_{2}(n):$}\\
 It is easy to note from the formula for
 $A_{B},B_{B},C_{B},D_{B},E_{B},F_{B},H_{B},$ and $I_{B}$, that  $A_{B}$ can be computed in $O(d_{T}(B)^{3})$ time,
  $B_{B}$ can  be computed in $O(d_{T}(B)^{2})$ time,
  $E_{B}$ can be computed in $O(d_{T}(B)^{3})$ time, $D_{B}$ can be computed in $O(d_{T}(B)^{2})$ time,
   $C_{B}$ can be computed in $O(d_{T}(B))$ time,
   $F_{B}$ can be computed in $O(1)$ time, $H_{B}$ can be computed in $O(1)$ time,
   and  $I_{B}$ can be computed in $O(1)$ time. Hence $T_{2}(n)=O(d_{T}(B)^{3})=O(n^{3})$.

 Therefore $T_{1}(n)+T_{2}(n)=O(n^{3})$, and hence $T(n)=O(n^{4})$.

 Hence our algorithm can be implemented in $O(n^{4})$ time and the correctness of the algorithm follows from the recurrences defined in above subsections. Therefore, the theorem is true.
\end{proof}

Note that a minimum cardinality restrained dominating set of a block
graph $G$ can be computed in $O(n^4)$ time by maintaining the sets
corresponding to the computed parameters at every node of the
refined cut-tree $T$ of $G$.

\section{Restrained domination in threshold graphs}
\label{sec:7} In this section we give a method to compute a minimum restrained dominating set of a threshold graph $G$ in linear-time.  A graph $G=(V,E)$ is called a \emph{threshold graph} if there is a real number $T$ and a real number $w(v)$ for every $v\in V$ such that a set $S\subseteq V$ is independent if and only if $\sum_{v\in S}w(v)\leq T$~\cite{tec_report}. Many characterizations of threshold graphs are available in the literature. An important characterization of threshold graph, which is used in designing polynomial time algorithms is following:

A graph $G$ is threshold graph if and only if it is a split graph and, for any split partition $(C,I)$ of $G$, there is an ordering $(x_{1},x_{2},\ldots,x_{p})$ of the vertices of $C$ such that $N_{G}[x_{1}]\subseteq N_{G}[x_{2}] \subseteq \cdots \subseteq N_{G}[x_{p}]$, and there is an ordering $(y_{1},y_{2},\ldots,y_{q})$ of the vertices of $I$ such that $N_{G}(y_{1}) \supseteq N_{G}(y_{2}) \supseteq \cdots \supseteq N_{G}(y_{q})$~\cite{threshold}.

\begin{theorem}
Let $G=(V,E)$ be a threshold graph having at least three vertices with split partition $(C,I)$ as defined above, then a minimum restrained dominating set $D_{r}^{*}$ of $G$ can be computed in the following way:
\[D_{r}^{*}=
\begin{cases}
        V, & if \hspace*{.2cm} p==1\\
       \{x_{p}\}, &  if \hspace*{.2cm} p>1 \hspace*{.1cm} and \hspace*{.1cm} N_{G}[x_{p}]=N_{G}[x_{p-1}]\\
     \{x_{p}\}\cup \{v\in I \mid v\in N_{G}(x_{p})\setminus N_{G}(x_{p-1})\},  &  if \hspace*{.2cm} p>1 \hspace*{.1cm} and \hspace*{.1cm} N_{G}[x_{p}]\neq N_{G}[x_{p-1}]
\end{cases}
\]
\end{theorem}
\

\begin{proof}
We know that a restrained dominating set must contain all the pendant vertices of graph. If $|C|=1$, then the only
non-pendant vertex is the vertex in the partite set $C$, say $v_{c}$. Then $I$ must be contained in $D^{*}$, where $D^{*}$ is any restrained dominating set of $G$. Since $v_{c}$ is an isolated vertex in $G[V\setminus I]$, by the definition of restrained dominating set, the vertex $v_{c}$ should also belongs to $D^{*}$. Thus $V=\{v_{c}\}\cup I$ is the only restrained dominating set of $G$ in the case when $|C|=1$. Hence $D^{*}=V$.

If $|C|>1$ and $N_{G}[x_{p}]=N_{G}[x_{p-1}]$, then $D^{*}=\{x_{p}\}$ is a dominating set of $G$. Note that $D^{*}$ is also a restrained dominating set of $G$ since every vertex in $V\setminus D^{*}$ except $v_{p-1}$ is adjacent to the vertex $v_{p-1}$. Also $|D^{*}|=1$, and hence $D^{*}=\{x_{p}\}$ is a minimum restrained dominating set of $G$.

Now consider the case when $|C|>1$ and $N_{G}[x_{p}] \neq N_{G}[x_{p-1}]$. Let $D_{r}^{*}$ be a minimum restrained
dominating set of $G$. The set $S=N_{G}(x_{p})\setminus N_{G}(x_{p-1})$ can be dominated by the vertex $x_{p}$ or the set $S$ itself. Hence either $x_{p}\in D_{r}^{*}$ or $S\subseteq D_{r}^{*}$.

Since all the vertices of $S$ are isolated in $G[V\setminus \{x_{p}\}]$, if $x_{p}\in D_{r}^{*}$, then $S$ must be contained in $D_{r}^{*}$. Hence $\{x_{p}\}\cup S \subseteq D_{r}^{*}$. But since $\{x_{p}\}\cup S$ is a restrained dominating set of $G$, $D_{r}^{*}=\{x_{p}\}\cup S$.

If $S \subseteq D_{r}^{*}$, then since $S$ does not dominate the vertex $x_{p-1}$, $|D_{r}^{*}|\geq |S|+1$. Hence $\{x_{p}\}\cup S$ is a minimum restrained dominating set of $G$. This completes the proof of the theorem.
 \end{proof}

Given a split partition $(C,I)$ of threshold graph $G$ and an
ordering $(x_{1},x_{2},\ldots,x_{p})$ of the vertices of $C$ such
that $N_{G}[x_{1}]\subseteq N_{G}[x_{2}] \subseteq \cdots \subseteq
N_{G}[x_{p}]$, and  an ordering $(y_{1},y_{2},\ldots,y_{q})$ of the
vertices of $I$ such that $N_{G}(y_{1}) \supseteq N_{G}(y_{2})
\supseteq \cdots \supseteq N_{G}(y_{q})$, one can compute in
$O(n+m)$ time a minimum cardinality restrained dominating set of a
threshold graphs.

\section{Restrained domination in cographs}
\label{sec:8} In this section, we show that the \textsc{Minimum Restrained Domination} problem can also be solved in linear time for cographs,
which is a super class of threshold graphs. A \emph{cograph} is a graph without induced $P_{4}$~\cite{corneil}.
Various characterizations of cographs are known in literature.
\begin{theorem}
A graph is a cograph if and only if every induced subgraph $H$ is disconnected or the complement $\overline{H}$ is disconnected.
\end{theorem}

A cograph has a tree decomposition which is called a cotree. A cotree is a pair $(T,f)$ comprising a rooted binary tree $T$ together with a bijection $f$ from from the vertices of the graph to the leaves of the tree. Each internal node of $T$ has a label $\otimes$ or $\oplus$. The operator $\otimes$ is called a join operation and it makes every vertex that is mapped to a leaf in the left subtree adjacent to every vertex that is mapped to leaf in the right subtree. The operator $\oplus$ is called union operation. In that case the graph is the union of the graphs defined by the left and right subtree. If n denotes the number of vertices in the graph, its cotree has $O(n)$ nodes. A cotree decomposition of a graph can also be obtained in linear-time~\cite{perl}. Suppose $I(G)$ denotes the number of isolated vertices in a graph $G$. The following result regarding the domination number of a cograph is already proved.
\begin{theorem}
\cite{poon}~If $G$ is a cograph with at least two vertices, then
\[\gamma(G)=
\begin{cases}
     \gamma(G_{1})+\gamma(G_{2}), &  if \hspace*{.2cm} G=G_{1} \oplus  G_{2} \\
     min\{\gamma(G_{1}),\gamma(G_{2}),2\},  &  if \hspace*{.2cm} G=G_{1} \otimes  G_{2}
\end{cases}
\]
\end{theorem}
Now we are ready to prove the following theorem.
\begin{theorem}
If $G$ is a cograph with at least two vertices, then
\[\gamma_{r}(G)=
\begin{cases}
     \gamma_{r}(G_{1})+\gamma_{r}(G_{2}), &  if \hspace*{.2cm} G=G_{1} \oplus  G_{2} \\
      2,  &   if\hspace*{.2cm}  G=G_{1} \otimes G_{2}, |V(G_{1})|=1, |V(G_{2})|= 1\\
    min\{\gamma(G_{1}),\gamma(G_{2}),2\},           &   if\hspace*{.2cm}  G=G_{1} \otimes G_{2},\hspace*{.2cm} |V(G_{1})|\geq 2,\hspace*{.2cm} |V(G_{2})|\geq 2\\
    min\{1+I(G_{2}),\gamma(G_{2})\}             &   if \hspace*{.2cm} G=G_{1} \otimes G_{2}, |V(G_{1})|= 1,  |V(G_{2})|\geq 2\\
    min\{1+I(G_{1}),\gamma(G_{1})\}             &   if \hspace*{.2cm} G=G_{1} \otimes G_{2}, |V(G_{1})|\geq 2, |V(G_{2})|= 1
\end{cases}
\]

\end{theorem}
\begin{proof}
 For the case $G=G_{1} \oplus  G_{2}$, clearly $\gamma_{r}(G)=\gamma_{r}(G_{1})+\gamma_{r}(G_{2})$, since no vertex of $G_{1}$ is adjacent to any vertex of $G_{2}$.

When $G=G_{1} \otimes G_{2}$ and $|V(G_{1})|=|V(G_{2})|= 1$, then $G=K_{2}$ and hence $\gamma_{r}(G)=2$.

Now consider the case when $G=G_{1} \otimes G_{2},$ $|V(G_{1})|\geq 2$, $|V(G_{2})|\geq 2$. Here $\{x,y\}$ is a restrained dominating set of $G$ where $x\in G_{1}$, $y\in G_{2}$. If there exists a vertex $x\in G_{1}$ which dominates all the vertices of $G_{1}$, then $\{x\}$ is a restrained dominating set of $G_{1}$. Similarly if there exists a vertex $y\in G_{2}$ which dominates all the vertices of $G_{2}$, then $\{y\}$ is a restrained dominating set of $G_{2}$. This proves the formula in this case.

Now consider the case when $G=G_{1} \otimes G_{2},$ $|V(G_{1})|=1$, $|V(G_{2})|\geq 2$. Let $V(G_{1})=\{v\}$.
Suppose $D_{r}^{*}$ is a minimum restrained dominating set of $G$. If $v\in D_{r}^{*}$, then $D_{r}^{*}=\{v\}\cup I(G_{2})$. If $v\notin D_{r}^{*}$, then $D_{r}^{*}$ is a minimum dominating set of $G_{2}$. This proves the formula in this case.

The formula for the case $G=G_{1} \otimes G_{2},$ $|V(G_{1})|\geq 2$, $|V(G_{2})|=1$ can be proved similarly.
\end{proof}

Using the above theorem, a minimum restrained dominating set of a
co-graph can be computed in $O(n+m)$ time.

\section{Restrained domination in chain graphs}
\label{sec:9}
A bipartite graph $G = (X, Y,E)$ is called a \emph{chain graph} if the
neighborhoods of the vertices of $X$ form a chain, that is, the
vertices of $X$ can be linearly ordered, say $x_1,x_2,\ldots,x_p$,
such that $N_G(x_1) \subseteq N_G(x_2) \subseteq \ldots \subseteq
N_G(x_p)$.  If $G = (X, Y,E)$ is a chain graph, then the
neighborhoods of the vertices of $Y$ also form a chain \cite{yan}.
An ordering $\alpha=(x_1, x_2,\ldots,x_p,y_1,y_2,\ldots,y_q)$ of $X
\cup Y$ is called a chain ordering if $N_G(x_{1})\subseteq
N_G(x_{2})\subseteq \cdots\subseteq N_G(x_{p})$ and
$N_G(y_{1})\supseteq N_G(y_{2})\supseteq \cdots \supseteq
N_G(y_{q})$. It is well known that every chain graph admits a chain
ordering \cite{yan,kloks}.

\begin{lemma}\label{lem1}
Let $G$ be a connected chain graph as defined above and $v$ be a pendant vertex of $G$. If $v\in Y$, then $v$ is  adjacent to $x_{p}$ and if $v\in X$, then $v$ is  adjacent to $y_{1}$.
\end{lemma}
\begin{proof}
By the definition of chain ordering, every vertex $y$ in $Y$ is adjacent to $x_{p}$, similarly every vertex $x$ in $X$ is adjacent to $y_{1}$. Hence the lemma is proved.
\end{proof}

\begin{lemma}\label{lem2}
Let $G$ ba a chain graph with at least three vertices. If every non-pendant vertex of $G$ is adjacent to
a pendant vertex as well as a non-pendant vertex, then $G$ is a bi-star.
\end{lemma}
\begin{proof}
If a vertex $v$ in $G$ is adjacent to a pendant vertex then either $v=x_{p}$ or $v=y_{1}$. Since every non-pendant vertex of $G$
is adjacent to a pendant vertex, $G$ has at most two non-pendant vertices. To show that $G$ is a bi-star, we need to show that
$G$ has exactly two non-pendant vertices. Clearly $G$ has at least one non-pendant vertex, say $v$. By the statement of lemma, $v$ must have
an adjacent non-pendant vertex. Therefore $G$ contains exactly two non-pendant vertices.
\end{proof}

\begin{theorem}\label{chain}
Let $G=(X,Y,E)$ be a connected chain graph having at least three vertices and
$\alpha=(x_1,x_2,\ldots,x_p,y_1,y_2, \ldots,y_q)$ is chain ordering
of $X \cup Y$.  Then $t \leq \gamma_{c}(G) \leq t+2$, where $t$ denotes
the number of pendant vertices of $G$. Furthermore, the following
are true.
\begin{enumerate}
  \item[(a)] $\gamma_{r}(G) = t$ if and only if $G=K_{2}$ or bi-star.
  \item[(b)] Let $P$ denotes the set of all pendant vertices of $G$ and $P_{A}$ denotes
  the set of vertices adjacent to the vertices of $P$. Then $\gamma_{r}(G)=t+1$ if and only if either $G$ or $G'=G[(X\cup Y)\setminus (P\cup P_{A})]$ is a star.
  \item[(c)] If $G$ is a graph other than the graphs described in the above statements then $\gamma_{r}(G)=t+2$.
\end{enumerate}
\end{theorem}
\begin{proof}
Since every restrained dominating set contains all the pendant vertices of graph,  $\gamma_{r}(G)\geq t$.
Now suppose that $P$ denotes the set of all pendant vertices of graph $G$. Then the set $D=P\cup \{x_{p},y_{1}\}$ is a
dominating set of $G$. We claim that $D$ is also a restrained dominating set of $G$. If not, then there exists a vertex
$v\in (X\cup Y)\setminus D$ such that $N_{G}(v)\subseteq D$. Since $|N_{G}(v)| \geq 2$, at least one neighbor of $v$ is
a pendant vertex. But then by Lemma~\ref{lem1}, $v$ is either $x_{p}$ or $y_{1}$, which is a contradiction.
Hence $D=P\cup \{x_{p},y_{1}\}$ is a restrained dominating set of $G$. This proves that $\gamma_{r}(G)\leq t+2$.

 $(a)$ If $G=K_{2}$ or bi-star, then $D=P$ is a restrained dominating set of $G$. Hence $\gamma_{r}(G)=t$.

Conversely suppose that $\gamma_{r}(G)=t$. Let $D$ be the minimum restrained dominating set of $G$. Then $D$ is exactly the set of all pendant vertices of $G$. Then we have two possibilities for graph $G$.

\noindent (i) All the vertices of graph $G$ are pendant vertices. Then $G=K_{2}$.

\noindent (ii) Every non-pendant vertex is adjacent to a pendant vertex as well as a non-pendant vertex. By Lemma~\ref{lem2}, chain graph satisfying this property is bi-star.

$(b)$ If $G$ is star, then clearly $\gamma_{r}(G)=|X\cup Y|=t+1$. If $G'$ is star with star center $v$, then $D=P\cup \{v\}$ is a dominating
set of $G$. Also $x_{p},y_{1}\notin D$, and every vertex $u$ not in $D$ is either adjacent to $x_{p}$ or $y_{1}$. Hence $D$ is also a restrained
dominating set of $G$ and $\gamma_{r}(G)=t+1$.

Conversely suppose that $\gamma_{r}(G)=t+1$. Let $D$ be the minimum restrained dominating set of $G$. Then $P\subseteq D$. Now we have two possibilities for graph $G$.

\noindent (i) The set $P$ dominates $X\cup Y$. But there exists a vertex $v\in (X \cup Y)\setminus P$ such that $N_{G}(v)\subseteq D$. Thus $v$ is a non-pendant vertex in the graph $G$ and all the neighbors of $V$ are pendant. This is possible only when $v$ is the only non-pendant vertex in the graph $G$. Hence $G$ is a star in this case.

\noindent (ii) The set $P$ does not dominate $X\cup Y$, but there exist a non-pendant vertex $v$ such that $P\cup\{v\}$ dominates all the vertices of the graph. Hence $v$ dominates the set $(X\cup Y)\setminus (P \cup P_{A})$. This is possible only when $G[(X\cup Y)\setminus (P \cup P_{A})]$ is a star.

$(c)$ Proof directly follows from above statements.
\end{proof}

Now we are ready to prove the following theorem:
\begin{theorem}
A minimum restrained dominating set of a chain graph can be computed in $O(n+m)$ time.
\end{theorem}
\begin{proof}
For a chain graph $G=(X,Y,E)$, a chain ordering $\alpha=(x_1,x_2,\ldots,x_p,y_1,y_2, \ldots,y_q)$ of $X\cup Y$ can be computed in linear time.
The set $P$ of all pendant vertices of $G$ can also be computed in linear time. Now, if $|X\cup Y|=2$, then take $D=X\cup Y$. It can also be computed in
linear time whether $G$ is a star or bistar. If $G$ is a bistar, then take $D=P$. If $G$ is a star with star center $v$, then take
$D=P\cup \{v\}$. Let $S$ denote
the set of vertices adjacent to a pendant vertex of $G$. Then the set $S$ can also be computed in linear time.
If $G'=G[(X \cup Y)\setminus (P \cup S)]$ is a star with star center $u$, then $D=P\cup \{u\}$. Otherwise, take $D=P\cup \{x_{p},y_{1}\}$. By
Theorem~\ref{chain}, $D$ is minimum cardinality restrained dominating set of $G$. Hence, the theorem is proved.
\end{proof}

\section{An upper bound for the restrained domination number}
\label{sec:10}
Zverovich and Pohosyan~\cite{zverovich} proved the following result.
\begin{theorem}
If a graph $G$ with $n$ vertices and minimum degree $\delta$ has a perfect matching, then
$$\gamma_{r}(G)\leq \frac{2(1+\ln(\delta+1))}{\delta+1}n+\epsilon$$  where $\epsilon=0$ if $n$ is even and $\epsilon=1$ otherwise.
\end{theorem}
In this section we prove the following stronger result for the restrained domination number of a graph $G$.
\begin{theorem}\label{rand_th}
Let $G$ be a connected graph with $n$ vertices and minimum degree $\delta$, then $G$ has a restrained dominating set of cardinality at most $\displaystyle\frac{ 2(1+\ln(\delta+1))}{\delta+1}n$.
\end{theorem}
\begin{proof}
Let $p\in [0,1]$ be arbitrary. We pick randomly and independently each vertex of graph $G$ with probability $p$. Let $A$ be the set of all picked vertices, $B_{A}$ be the random set of all vertices in $V\setminus A$ that do not have any neighbor in $A$, and $C_{A}$ be the set of all vertices in $V\setminus (A\cup B)$ for which all the neighbors are in $A\cup B$. Then $E(|A|)=np$.

 For each fixed vertex $v\in V$, $Pr(v\in B_{A})=Pr(v$ and all its neighbors are not in $A)\leq (1-p)^{\delta+1}$. Hence $E(|B_{A}|)\leq n(1-p)^{\delta+1}$.

$$Pr(v\in A\cup B_{A})=Pr(v\in A)+Pr(v\in B_{A})\leq p+(1-p)^{\delta+1}$$
$$Pr(v\notin A\cup B_{A})=1-Pr(v\in A\cup B_{A})=1-(Pr(v\in A)+Pr(v\in B_{A}))\leq  1-p$$

 For each fixed vertex $v\in V$, $Pr(v\in C)=Pr(v\notin A\cup B_{A}$ and all its neighbors are in $A\cup B_{A})\leq (1-p)(p+(1-p)^{\delta+1})^{\delta}$. Hence $E(|B_{A}|)\leq n(1-p)(p+(1-p)^{\delta+1})^{\delta}$.

 Now
 \begin{eqnarray}
 \nonumber  E(|A|+|B_{A}|+|C_{A}|)& \leq & np+n(1-p)^{\delta+1}+n(1-p)(p+(1-p)^{\delta+1})^{\delta}\\
\nonumber & \leq & np+n(1-p)^{\delta+1}+n(1-p)(p+(1-p)^{\delta+1})\\
\nonumber & \leq & np+n(1-p)^{\delta+1}+np(1-p)+n(1-p)^{\delta+2}\\
\nonumber & \leq & np+n(1-p)^{\delta+1}+np+n(1-p)^{\delta+1}\\
\nonumber & \leq & 2np+2n(1-p)^{\delta+1}\\ 
\nonumber & \leq & 2np+2n e^{-p(\delta+1)} \hspace*{1.1cm} (\because \hspace*{.2cm} 1-p\leq e^{-p})
\end{eqnarray}

Hence, there is a at least one choice of $A$ such that $|A|+|B_{A}|+|C_{A}|\leq  2np+2n e^{-p(\delta+1)}$. Also, the set $D_{r}=A\cup B_{A} \cup C_{A}$ is a restrained dominating set of $G$ and $|D_{r}|=|A|+|B_{A}|+|C_{A}|$. Hence there exist at least one restrained dominating set, say $D_{r}$ of $G$ such that $|D_{r}|\leq 2np+2n e^{-p(\delta+1)}$. Thus $\gamma_{r}(G)\leq 2np+2n e^{-p(\delta+1)}$. This holds for any $p\in [0,1]$.

Now to find the minimum value of the bound, we differentiate the right hand side with respect to $p$ and set it equal to zero. The minimum value of right hand side will be obtained at $p=\frac{\ln(\delta+1)}{\delta+1}$. When we put this value of $p$ in the upper bound, we get $\gamma_{r}(G)\leq \displaystyle\frac{ 2(1+\ln(\delta+1))}{\delta+1}n$.
\end{proof}

Now we present a randomized algorithm to find a restrained dominating set $D_{r}$ of graph $G$. The algorithm is based on the probabilistic construction used in Theorem~\ref{rand_th}. The expectation of cardinality of restrained dominating set $D_r$ returned by the following algorithm satisfies the upper bound of Theorem~\ref{rand_th}. Note that the algorithm can be implemented in linear-time.

\begin{algorithm}[H]
\caption{RANDOMIZED-RESTRAINED-DOM-SET(G)}
 \textbf{Input:} A graph $G=(V,E)$.\\
\textbf{Output:} A restrained dominating set $D_{r}$ of $G$.\\
\Begin{
Compute $p=\frac{\ln(\delta+1)}{\delta+1}$;\\
Initialize $A=B_{A}=C_{A}=\emptyset$;\\
\ForEach { $v\in V(G)$ }
{ with the probability $p$ decide if $v\in A$ or $v\notin A$; }
\ForEach {$v\in V(G)\setminus N_{G}[A]$ }
{$B_{A}=B_{A}\cup \{v\}$;}
\ForEach {$v\in V(G)\setminus (A\cup B)$ }
{
\If {$N_{G}(v)\subseteq A\cup B$}
{$C_{A}=C_{A}\cup \{v\}$;}
}
Put $D_{r}=A\cup B_{A} \cup C_{A}$;\\
return $D_{r}$;
}
\end{algorithm}

\section{Conclusion}
\label{sec:11}
In this paper, we studied algorithmic aspects of the MRD problem.
The RDD problem is already known to be NP-complete for chordal graphs,
split graphs, planar graphs, undirected path graphs, and bipartite graphs. On the positive side, polynomial time algorithms
are known to solve the MRD problem in trees and proper interval graphs. We proved that
RDD problem is NP-complete for doubly chordal graphs, a subclass of chordal graphs, and
proposed a polynomial time algorithm to solve the MRD problem
in block graphs, a subclass of doubly chordal graphs. We also proposed algorithms to solve the MRD problem in threshold
graphs (subclass of split graphs), cographs (superclass of threshold graphs), and chain graphs (subclass of bipartite graphs).
We also observed that there exist graph classes for which domination and restrained domination differ in complexity.
It is still interesting to look at the complexity status of the problem for other important subclasses of bipartite graphs and chordal graphs.
In addition, We provided an upper bound on the restrained domination number of a graph in terms of number of vertices and degree of graph.
We also proposed a randomized algorithm to compute a restrained dominating set of a graph, and proved that the cardinality of the
restrained dominating set returned by our algorithm satisfies our upper bound with a positive probability.

\end{document}